\documentclass[10pt,twocolumn,twoside]{IEEEtran}

\IEEEoverridecommandlockouts                              	
\usepackage{amssymb,amsmath,amsfonts,amsthm}
\usepackage{algorithm,algorithmic}
\usepackage{cite}
\usepackage{float}
\usepackage{epsfig}
\usepackage{graphicx}
\usepackage{graphics}
\usepackage{subfigure}
\usepackage{url}
\usepackage{xspace}
\usepackage[usenames,dvipsnames]{color}
\usepackage{soul}
\usepackage{mathtools}
\floatstyle{ruled}
\newfloat{model}{H}{mod}
\floatname{model}{\footnotesize Model}
\newfloat{notatio}{H}{not}
\floatname{notatio}{\footnotesize Notation}

\newenvironment{varalgorithm}[1]
  {\algorithm\renewcommand{\thealgorithm}{#1}}
  {\endalgorithm}

\newenvironment{list3}{
	\begin{list}{$\bullet$}{%
			\setlength{\itemsep}{0.05cm}
			\setlength{\labelsep}{0.2cm}
			\setlength{\labelwidth}{0.3cm}
			\setlength{\parsep}{0in} 
			\setlength{\parskip}{0in}
			\setlength{\topsep}{0in} 
			\setlength{\partopsep}{0in}
			\setlength{\leftmargin}{0.22in}}}
	{\end{list}}

\newenvironment{list4}{
	\begin{list}{$\bullet$}{%
			\setlength{\itemsep}{0.05cm}
			\setlength{\labelsep}{0.2cm}
			\setlength{\labelwidth}{0.3cm}
			\setlength{\parsep}{0in} 
			\setlength{\parskip}{0in}
			\setlength{\topsep}{0in} 
			\setlength{\partopsep}{0in}
			\setlength{\leftmargin}{0.16in}}}
	{\end{list}}

\newenvironment{list4a}{
	\begin{list}{$\bullet$}{%
			\setlength{\itemsep}{0.05cm}
			\setlength{\labelsep}{0.2cm}
			\setlength{\labelwidth}{0.3cm}
			\setlength{\parsep}{0in} 
			\setlength{\parskip}{0in}
			\setlength{\topsep}{0in} 
			\setlength{\partopsep}{0in}
			\setlength{\leftmargin}{0.16in}}}
	{\end{list}}

\usepackage{url,changebar,bm,xspace,dsfont}
\let\mathbb=\mathds % I much prefer the dsfont over the bbfont
\newtheorem{theorem}{Theorem}

\newtheorem{defn}{Definition}

\newtheorem{assum}{Assumption}

\newtheorem{remark}{Remark}

\ifCLASSINFOpdf
 \else
\fi

\usepackage{amsmath,amsfonts,amssymb,amscd}
\usepackage{capt-of}
\usepackage{color}
\usepackage{cite}
\usepackage{verbatim}
\usepackage{hyperref}

\hypersetup{
    colorlinks = true,
    linkcolor = [rgb]{0,0,1},
    anchorcolor = [rgb]{0,0,1},
    citecolor = [rgb]{0.9,0.5,0},
    filecolor = [rgb]{0,0,1},
    urlcolor = [rgb]{0.9,0.5,0},
    bookmarksopen = true,
    bookmarksnumbered = true,
    breaklinks = true,
    linktocpage,
    colorlinks = true,
    linkcolor = [rgb]{0.2,0.6,0.2},
    urlcolor  = [rgb]{0,0,1},
    citecolor = [rgb]{0.9,0.5,0},
    anchorcolor = [rgb]{0.2,0.6,0.2},
}
\newtheorem{lemma}{\bfseries Lemma}

\begin{document}

%\title{\LARGE \bf Optimal CPU Scheduling in Data Centers \\ via a Finite Time Communication Efficient Algorithm}
\title{\LARGE \bf Distributed Computation of Exact Average Degree and Network Size \\ in Finite Number of Steps under Quantized Communication}

%\author{ \parbox{3 in}{\centering Huibert Kwakernaak*
%         \thanks{*Use the $\backslash$thanks command to put 
%information here}\\
%         Faculty of Electrical Engineering, Mathematics and 
%Computer Science\\
%         University of Twente\\
%         7500 AE Enschede, The Netherlands\\
%         {\tt\small h.kwakernaak@autsubmit.com}}
%         \hspace*{ 0.5 in}
%         \parbox{3 in}{ \centering Pradeep Misra**
%         \thanks{**The footnote marks may be inserted manually}\\
%        Department of Electrical Engineering \\
%         Wright State University\\
%         Dayton, OH 45435, USA\\
%         {\tt\small pmisra@cs.wright.edu}}
%}

\author{Apostolos~I.~Rikos, Themistoklis Charalambous, Christoforos N. Hadjicostis, and Karl~H.~Johansson
\thanks{Apostolos~I.~Rikos and K.~H.~Johansson are with the Division of Decision and Control Systems, KTH Royal Institute of Technology, SE-100 44 Stockholm, Sweden. E-mails: {\tt \{rikos,kallej\}@kth.se}.}
\thanks{{T.~Charalambous and} C.~N.~Hadjicostis {are} with the Department of Electrical and Computer Engineering, University of Cyprus, 1678 Nicosia, Cyprus.  E-mails:{\tt \{charalambous.themistoklis,chadjic\}@ucy.ac.cy}.}
}

\maketitle
\thispagestyle{empty}
\pagestyle{empty}

% ===============================================
%
%
% ABSTRACT
%
%
% ===============================================
\begin{abstract}
We consider the problems of computing the average degree and the size of a given network in a distributed fashion under quantized communication. 
%Computing the average {outgoing} degree involves determining the average number of edges from each node. This can be useful in various scenarios, such as infection propagation and antidote distribution to control epidemics. 
%Computing the network size involves determining in a distributed way the number of nodes in a network. Knowing the network size {is} useful in several scenarios such as load balancing and detection of topological changes. 
We present two distributed algorithms which rely on quantized operation (i.e., nodes process and transmit quantized messages), and are able to calculate the exact solutions in a finite number of steps. 
Furthermore, during the operation of our algorithms, each node is able to determine in a distributed manner whether convergence has been achieved and thus terminate its operation.
To the best of the authors' knowledge these algorithms are the first to find the exact solutions under quantized communication (i.e., there is no error in the final calculation).
Additionally, note that our network size calculation algorithm is the first in the literature which calculates the exact size of a network in a finite number of steps without introducing a final error. This error in other algorithms can be either due to quantization or asymptotic convergence. In our case, no error is introduced since the desired result is calculated in the form of a fraction involving a quantized numerator and a quantized denominator.  
%We analyze the operation of our algorithms, and prove that they calculate the desired values. {Note that when a leader is not already assigned, the algorithms converge with high probability, since the leader election process is successful with high probability.} 
Finally, we demonstrate the operation of our algorithms and their potential advantages.
\end{abstract}

% \begin{IEEEkeywords} 
% \todo{FIX} 
% \end{IEEEkeywords}

% ===============================================
%
%
% INTRODUCTION
%
%
% ===============================================
\section{Introduction}\label{intro}

%In recent years distributed control and coordination of multi-agent systems has attracted a lot of interest from the research community due to their applicability in various scenarios \cite{2008:Cortes, 2012:Varagnolo_Schenato}. 
Many distributed algorithms which operate over multi-agent systems rely on knowledge of the network's parameters, such as the average degree and/or the size of the network. 
Knowledge of the average node degree or the size of the network plays an important role in various applications, including consensus based distributed optimization \cite{Rabbat:2018IEEEProceedings}, infection propagation strategies \cite{2007:Leskovec_Faloutsos}, antidote distribution to control epidemics \cite{2010:Borgs_Saberi}, and the networked prisoner's dilemma game \cite{2010:Tang_Wang}. 
Furthermore, knowing the network's parameters is important for detecting topological changes \cite{2014:Lucchese_Johansson}, estimating the maximum and the minimum of the initial measurements in the presence of noise via a soft-max operation \cite{2013:Zhang_Spanias}, 
% distributed clustering over wireless networks \cite{2015:Oliva}, 
and controlling renewable energy resources, while maintaining an average degree that fulfils structural properties \cite{2010:christoforos}.

Various methods have been proposed in the literature for calculating the size of a given network. 
Current approaches rely on statistical methods that require the exchange of excessive information between nodes, random walk strategies, random sampling, and capture-recapture strategies \cite{2010:Varagnolo_Schenato, 2012:Garin_Johansson, 2006:Massoulie_Ganesh, 2009:Peng_Xiao, 2015:Lucchese_Hendrickx, 2019:Kenyeres_Kenyeres, 2021:Kenyeres_Budinska}. 
Furthermore, the problem of calculating the average degree of a given network has been analyzed in \cite{2012:Ribeiro_Towsley, 2009:Hay_Jensen, 2014:Dasgupta_Sarlos}; however, finding its solution in a distributed fashion has received limited attention. 
Specifically, only \cite{2012:Shames_Johansson} presents an asymptotic distributed algorithm for computing the average degree of a network. 
%This approach requires each node to transmit real valued messages of infinite precision (which significantly increases bandwidth requirements), and exhibits asymptotic convergence which introduces a final error on the calculation of the desired result. 
To the best of the authors' knowledge, every algorithm in the literature which calculates the size and the average degree of a given network operates with real values and exhibits asymptotic convergence. 
However, operation with real values in practice requires channels with infinite bandwidth (since an irrational number requires an infinite number of bits for representation). 
Therefore, nodes need to approximate such numbers, but still require high bandwidth.
This is undesired in practical applications, since the communication overhead of each node becomes the major operational bottleneck over large scale networks. 
Furthermore, asymptotic convergence introduces a final error on the calculated result because most algorithms need to be terminated after a pre-defined finite number of iterations. 
This error may lead to imprecise calculation of the desired quantity, which may be of significant magnitude for the case of large scale networks.  
% This leads to significant increments on each node's bandwidth requirements, and also introduces a final error on the calculated result{, since the algorithms are usually terminated in a finite number of steps}. 
%The aim of this paper is to address the aforementioned problems since the area of 
In summary, calculating network parameters in a finite number of steps with quantized communication remains largely unexplored.

% In practical applications (such as voting protocols, detection of topological changes, and infection propagation), imprecise calculations of the desired result may lead to 
% As a result, calculating the 

%\noindent
\textbf{Main Contributions.}
Our paper is a major departure from the current literature and aims to bridge this gap. 
Compared to existing approaches, the operation of our algorithms allows nodes to process and transmit quantized values. It is important to note that quantized operation allows more efficient usage of network resources compared to real valued operation, as nodes require less bits to store and transmit information. 
%Furthermore, our algorithms converge to the \textit{exact} result in a finite number of steps \TC{with high probability}. 
Furthermore, compared to algorithms that exhibit asymptotic convergence, our algorithms converge in a finite number of steps without introducing any final error, as nodes compute the final result in the form of a quantized fraction. 
%This means that current approaches require an infinite number of time steps to calculate the exact values of the parameters which our algorithms are able to calculate in finite time. 
%In this paper, we present two novel distributed algorithms for computing the average degree and the size of a given network, in a finite number of steps. 
%Their operation relies on processing and transmitting quantized values, and on event-driven state updates. 
The main contributions are the following. 
\begin{list4}
\item We present a novel distributed algorithm for computing the average degree of a given network. 
Our algorithm operates with quantized values and is able to calculate the exact result without any error; see Algorithm~\ref{algorithm_averagedegree}. 
\item We present a novel distributed algorithm for computing the size of a given network. 
Our algorithm operates with quantized values and calculates the exact result without introducing any final error; see Algorithm~\ref{algorithm_networksize}. 
Furthermore, the algorithm's operation relies on the election of a leader node (if a leader is not already assigned/decided). 
For this reason, we present a novel strategy for leader election with quantized processing and communication. 
Note that this is the first leader election strategy which relies on quantized operation; see Algorithm~\ref{algorithm_LeaderElection}. We show that the leader election strategy converges to the exact result (i.e., the election of one leader node) after a small number of time steps, with high probability; see Theorem~\ref{thm:size_estimation}. {Note that when a leader is not already assigned, the algorithm converges to the exact result with high probability, since the leader election process is successful with high probability.}
\item We show that both our algorithms converge in finite time.
For both algorithms, we provide upper bounds on the number of time steps needed for convergence. 
Our provided bounds rely on a known upper bound on the diameter rather than the size of the network. 
\item Both algorithms utilize a distributed stopping strategy with which nodes are able to determine whether convergence has been achieved. %It is the first method in the literature
%This is the fist distributed stopping strategy 
%which allows convergence to the exact solution without a final error. 
% \item We present simulations of our proposed algorithms where we show their finite time convergence to the exact solution. 
% Furthermore, we present simulations of the leader election strategy where we demonstrate its finite time converge; see Section~\ref{sec:results}. 
\end{list4}
% The main advantage of our algorithms is that they allow calculation of the \textit{exact} solution in the form of a quantized fraction without introducing a final error (either due to quantization or due to asymptotic convergence). 

%The main advantage of our proposed algorithms is that each node's state is represented as a quantized fraction. The numerator of the initial fraction is the node's initial state value and the denominator is equal to one. Then, it transmits the initial fraction to a randomly chosen neighbor node. If two or more fractions are transmitted to the same node, then the node sums the numerator and the denominator and forms a new fraction.  In this way, after a finite number of time steps each node will receive a fraction whose numerator is the sum of each node's initial state value and the denominator is equal to the number of nodes in the network. 

%In the current literature every algorithm operates with real values and exhibits asymptotic convergence within some error. 

%Specifically, for the case of computing the network size, since errors over the \todo{Compare with previous results} 

% ===============================================
%
%
% NOTATION
%
%
% ===============================================
\section{Notation and Preliminaries}\label{sec:preliminaries}

The sets of real, rational, and integer numbers are denoted by $ \mathbb{R}, \mathbb{Q}$, and $\mathbb{Z}$, respectively. 
The symbol $\mathbb{Z}_{\geq 0}$  ($\mathbb{Z}_{>0}$) denotes the set of nonnegative (positive) integer numbers (similarly  $\mathbb{Z}_{\leq 0}$, $\mathbb{Z}_{<0}$). 
%The symbols $\mathbb{Z}_{\leq 0}$ ($\mathbb{Z}_{<0}$) denote the sets of nonpositive (negative) integer numbers. 
% For any real number $a \in \mathbb{R}$, the floor $\lfloor a \rfloor$ denotes the greatest integer less than or equal to $a$ while the ceiling $\lceil a \rceil$ denotes the least integer greater than or equal to $a$. 
Vectors are denoted by small letters, matrices are denoted by capital letters and the transpose of a matrix $A$ is denoted by $A^T$. 
For a matrix $A\in \mathbb{R}^{n\times n}$, the entry at row $i$ and column $j$ is denoted by $A_{ij}$.
By $\mathbf{1}$ we denote the all-ones vector and by $I$ we denote the identity matrix (of appropriate dimensions). 

\textbf{Graph-Theoretic Notions.}
Consider a network of $n$ ($n \geq 2$) nodes communicating only with their immediate neighbors. 
The communication topology is captured by a directed graph (digraph) defined as $\mathcal{G}_d = (\mathcal{V}, \mathcal{E})$. 
In digraph $\mathcal{G}_d$, $\mathcal{V} =  \{v_1, v_2, \dots, v_n\}$ is the set of nodes, whose cardinality is denoted as $n  = | \mathcal{V} | \geq 2 $, and $\mathcal{E} \subseteq \mathcal{V} \times \mathcal{V} - \{ (v_j, v_j) \ | \ v_j \in \mathcal{V} \}$ is the set of edges (self-edges excluded) whose cardinality is denoted as $m = | \mathcal{E} |$. 
A directed edge from node $v_i$ to node $v_j$ is denoted by $m_{ji} \triangleq (v_j, v_i) \in \mathcal{E}$, and captures the fact that node $v_j$ can receive information from node $v_i$ (but not the other way around). 
We assume that the given digraph $\mathcal{G}_d = (\mathcal{V}, \mathcal{E})$ is \textit{strongly connected}. 
This means that for each pair of nodes $v_j, v_i \in \mathcal{V}$, $v_j \neq v_i$, there exists a directed \textit{path}\footnote{A directed \textit{path} from $v_i$ to $v_j$ exists if we can find a sequence of nodes $v_i \equiv v_{l_0},v_{l_1}, \dots, v_{l_t} \equiv v_j$ such that $(v_{l_{\tau+1}},v_{l_{\tau}}) \in \mathcal{E}$ for $ \tau = 0, 1, \dots , t-1$.} from $v_i$ to $v_j$. 
Furthermore, the diameter $D$ of a digraph is the longest shortest path between any two nodes $v_j, v_i \in \mathcal{V}$ in the network. 
The subset of nodes that can directly transmit information to node $v_j$ is called the set of in-neighbors of $v_j$ and is represented by $\mathcal{N}_j^- = \{ v_i \in \mathcal{V} \; | \; (v_j,v_i)\in \mathcal{E}\}$. 
The cardinality of $\mathcal{N}_j^-$ is called the \textit{in-degree} of $v_j$ and is denoted by $\mathcal{D}_j^-$. 
The subset of nodes that can directly receive information from node $v_j$ is called the set of out-neighbors of $v_j$ and is represented by $\mathcal{N}_j^+ = \{ v_l \in \mathcal{V} \; | \; (v_l,v_j)\in \mathcal{E}\}$. 
The cardinality of $\mathcal{N}_j^+$ is called the \textit{out-degree} of $v_j$ and is denoted by $\mathcal{D}_j^+$. 

\textbf{Node Operation.}
The operation of each node $v_j \in \mathcal{V}$ respects the quantization of information flow.  
At time step $k \in \mathbb{Z}_{>0}$, each node $v_j$ maintains the mass variables $y_j[k] \in \mathbb{Z}$ and $z_j[k] \in \mathbb{Z}_{\geq 0}$, which are used to communicate with other nodes. 
The state variables $y^s_j \in \mathbb{Z}$, $z^s_j \in \mathbb{Z}_{>0}$ and $q_j^s \in \mathbb{Q}$, (where $q_j^s = \frac{y_j^s}{z_j^s}$) are used to store the received messages. 
The voting variables $m_j$ and $M_j$ are used to determine whether convergence has been achieved and the nodes can stop their operation. 
Furthermore, we assume that each node $v_j$ is aware of its out-neighbors and can directly transmit messages to each out-neighbor separately. 
In order to randomly determine which out-neighbor to transmit to, each node $v_j$ assigns a nonzero probability $b_{lj}$ to each of its outgoing edges $v_l \in \mathcal{N}^+_j$. 
For every node, this probability assignment can be captured by an $n \times n$ column stochastic matrix $\mathcal{B} = [b_{lj}]$. 
A simple choice is to set these probabilities to be equal, i.e.,
\begin{align}\label{prob_assignment}
b_{lj} = \left\{ \begin{array}{ll}
         \frac{1}{1+ \mathcal{D}_j^+}, & \mbox{if $l = j$ or $v_{l} \in \mathcal{N}_j$,}\\
         0, & \mbox{otherwise.}\end{array} \right. 
\end{align}
Each nonzero entry $b_{lj}$ of matrix $\mathcal{B}$ represents the probability of node $v_j$ transmitting towards out-neighbor $v_l \in \mathcal{N}_j$. 
%Note here that during every $k$, we have $b_{lj}[k] = b_{jl}[k]$, for every $v_l, v_i \in \mathcal{V}$. 

%\section{Preliminaries on Distributed Coordination}\label{sec:DistributedCoordination}

\subsection{Synchronous $\max$/$\min$ - Consensus}

The $\max$-consensus algorithm computes the maximum value of the network in a finite number of time steps in a distributed fashion \cite{2008:Cortes}. 
For every node $v_{j} \in \mathcal{V}$, if the updates of the node's state are synchronous, then the update rule is:
%\begin{align}
$x_j[k+1] = \max_{v_{i}\in \mathcal{N}_j^{-} \cup \{v_{j}\}}\{ x_i[k] \}$.
%\end{align}
It has been shown (see, e.g., \cite[Theorem 5.4]{2013:Giannini}) that the $\max$-consensus algorithm converges to the maximum value among all nodes in a finite number of steps $s$, where $s \leq D$.  
Similar results hold for the $\min$-consensus algorithm.

\subsection{Quantized Average Consensus}\label{Prel_Aver}

The objective of quantized average consensus problems is the development of distributed algorithms which allow nodes to process and transmit quantized information. 
During their operation, each node utilizes short communication packages and eventually obtains after a finite number of time steps a state $q^s$ in the form of a quantized fraction, which is equal to the \textit{exact} real average $q$ of the initial states. 
Note that in this paper we consider the case where quantized values are represented by integer\footnote{Following \cite{2007:Basar} we assume that the state of each node is integer valued.
This abstraction subsumes a class of quantization effects (e.g., uniform quantization).} numbers. 
%This means that each node in the network is able to obtain a state $q^s$ which is equal to the ceiling $\lceil q \rceil$ or the floor $\lfloor q \rfloor$ of the real average $q$ of the initial quantized states of the nodes, after a finite number of time steps.

%\todo{??? Does every node only transmit its mass variable to only one of its out-neighbors or itself with probability? If the information transmition is with probability, the convergence is also with probability, which contradicts with the result that algorithms converge in finite time. - THEMIS WILL FIX}

Since each node processes and transmits quantized information, we adopt the algorithm in \cite{2018:RikosHadj}. 
This algorithm is preliminary for our results in this paper and during its operation, each node is able to achieve quantized average consensus after a finite number of time steps. 
The operation of the algorithm presented in \cite{2018:RikosHadj}, assumes that each node $v_j$ in the network has an integer initial state $y_j[1] \in \mathbb{Z}$. 
Initially, each node $v_j$ assigns a nonzero probability to each outgoing edge and to a virtual self-edge as \eqref{prob_assignment}.
At each time step $k$, each node $v_j \in \mathcal{V}$ maintains its mass variables $y_j[k], z_j[k]$, and its state variables $y^s_j[k], z^s_j[k], q_j^s[k]$. 
It updates the values of the mass variables as 
\begin{subequations}\label{Prel_Aver_YZ}
\begin{align}
y_j[k+1] = y_j[k] + \sum_{v_i \in \mathcal{N}_j^-} \mathds{1}_{ji}[k] y_i[k] , \label{subeq:1a} \\
z_j[k+1] = z_j[k] + \sum_{v_i \in \mathcal{N}_j^-} \mathds{1}_{ji}[k] z_i[k] , \label{subeq:1b}
\end{align}
\end{subequations}
where $\mathds{1}_{ji}[k] = 1$ if a message is received at $v_j$ from $v_i$ at $k$ ($0$ otherwise). 
% \begin{align*}
% \mathds{1}_{ji}[k] = 
% \begin{cases}
% 1, & \text{if a message is received at $v_j$ from $v_i$ at $k$,} \\[0.1cm]
% 0, & \text{otherwise.} 
% \end{cases}
% \end{align*}
If the following condition holds: 
\begin{list3}\label{tr_cond}
\item[(C1):] $z_j[k + 1] \geq 1$ ,
\end{list3}
then, node $v_j$ updates its state variables as  
\begin{subequations}\label{State_Prel_Aver_YZ}
\begin{align}
z^s_j[k+1] &= z_j[k+1], \\
y^s_j[k+1] &= y_j[k+1], \\
q^s_j[k+1] &= \frac{y^s_j[k + 1]}{z^s_j[k + 1]} .
\end{align}
\end{subequations}
Then, it  transmits its mass variables $z_j[k+1]$, $y_j[k+1]$ to one randomly selected out-neighbor or to itself according to \eqref{prob_assignment}. 
If it transmits its mass variables, it sets them equal to zero $z_j[k+1] = 0$, $y_j[k+1] = 0$.
Finally, it receives the values $y_i[k]$ and $z_i[k]$ from its in-neighbors $v_i \in \mathcal{N}_j^+$, {it executes the operations in \eqref{subeq:1a}, \eqref{subeq:1b}}, and repeats the operation. 

\begin{defn}\label{Definition_Quant_Av}
The system achieves \textit{exact} quantized average consensus if, for every $v_j \in \mathcal{V}$, there exists $k_0 \in \mathbb{Z}_{>0}$ so that for every $v_j \in \mathcal{V}$ we have
$
y^s_j[k] = \sum_{l=1}^{n}{y_l[1]} 
$
and 
$
z^s_j[k] = n ,
$
which means that 
% \begin{equation}\label{alpha_q_no_oscill}
$
q^s_j[k] = \frac{\sum_{l=1}^{n}{y_l[1]}}{n} = q , 
$
% \end{equation}
for $k \geq k_0$, where $q$ is the desirable (real) average of the initial states. 
% \begin{equation}\label{real_av}
% $
% q = \frac{\sum_{l=1}^{n}{y_l[1]}}{n} .
% $
% \end{equation}
\end{defn}
The following result from~\cite{2018:RikosHadj} analyzes the convergence of the quantized average consensus algorithm. 

\begin{theorem}[\hspace{-0.00001cm}\cite{2018:RikosHadj}]
\label{Conver_Quant_Av}
Consider a strongly connected digraph $\mathcal{G}_d = (\mathcal{V}, \mathcal{E})$ with $n=|\mathcal{V}|$ nodes and $m=|\mathcal{E}|$ edges and $z_j[1] = 1$ and $y_j[1] \in \mathbb{Z}$ for every node $v_j \in \mathcal{V}$ at time step $k=1$. 
Suppose that each node $v_j \in \mathcal{V}$ follows the Initialization and Iteration steps as described in the algorithm in \cite{2018:RikosHadj}. 
With probability one, we can find $k_0 \in \mathbb{N}$, so that for every $k \geq k_0$ we have 
$
y^s_j[k] = \sum_{l=1}^{n}{y_l[1]} 
$
and 
$
z^s_j[k] = n , 
$
which means that 
$
q^s_j[k] = \frac{\sum_{l=1}^{n}{y_l[1]}}{n} ,
$
for every $v_j \in \mathcal{V}$ (i.e., for $k \geq k_0$ every node $v_j$ has calculated $q$ as the ratio of two integer values). 
\end{theorem}

% ===============================================
%
%
% PROBLEM
%
%
% ===============================================
\section{Problem Formulation}\label{sec:probForm}

Consider a network modelled as a directed graph $\mathcal{G}_d = (\mathcal{V}, \mathcal{E})$. 
In this paper, we develop a distributed algorithm that allows nodes to address the problems $\textbf{P1}$ and $\textbf{P2}$ presented below, while processing and transmitting \textit{quantized} information via available communication links. 

$\textbf{P1.}$ After a finite number of time steps, each node $v_j$ obtains a fraction $q_j^s$, which is equal to  
\begin{equation}\label{average_degree_def}
    q_j^s = \frac{\sum_{l=1}^{n}{\mathcal{D}_l^+}}{n} \ ,
\end{equation}
which means that $q_j^s$ is equal to the average degree in the network. Each node $v_j$  processes and transmits quantized values, and ceases transmissions once \eqref{average_degree_def} holds for every node. 

$\textbf{P2.}$ After a finite number of time steps, each node $v_j$ obtains a value $z_j^s$, which is equal to: 
\begin{equation}\label{number_nodes_network}
    z_j^s = n \ , 
\end{equation}
which means that $z_j^s$ is equal to the number of nodes in the network. Each node $v_j$ processes and transmits quantized values, and 
ceases transmissions once \eqref{number_nodes_network} holds for every node.

% ===============================================
%
%
% ALGORITHM
%
%
% ===============================================
\section{Distributed Average Degree Computation}
\label{sec:distr_alg_averagedegree}

In this section we present a distributed algorithm which solves problem $\textbf{P1}$. 
Our algorithm is detailed below as Algorithm~\ref{algorithm_averagedegree}. 
%Algorithm~\ref{algorithm_averagedegree} operates solely with quantized values (i.e., each node processes and transmits quantized values) and converges in finite time \AR{with high probability}. 
For solving the problem in a distributed way we make the following assumption. 

\begin{assum}\label{Diam_known}
An upper bound $D'$ of the diameter of the network $D$ is known to all nodes $v_j \in \mathcal{V}$.
\end{assum}
Assumption~\ref{Diam_known} is necessary for coordinating $\min$- and $\max$-consensus algorithms, as it will be described later. 

We now describe the main operations of Algorithm~\ref{algorithm_averagedegree}. 
% The initialization involves the following steps: 

\textbf{Initialization.} 
Each node $v_j \in \mathcal{V}$ assigns a nonzero probability to each outgoing edge and a virtual self edge, so that the sum of the nonzero probabilities is equal to one. 
It sets its mass variable $y_j[1]$ to be equal to the node's out-degree and its mass variable $z_j[1]$ to be equal to one. 
Also it sets its initial state variables $y^s_j[1]$, $z^s_j[1]$ to be equal to the initial mass variables $y_j[1]$, $z_j[1]$, respectively, and the state variable $q^s_j[1]$ to be equal to the fraction $y^s_j[1] / z^s_j[1]$. 
Then, it transmits its mass variables to a randomly chosen out-neighbor and sets them equal to zero. 

% The iteration involves the following steps: 

\textbf{Iteration - Step~$1$. Calculating the Average Network Degree:} 
Every node $v_j$ receives the transmitted mass variables of its in-neighbors, and sums them with the stored mass variables. 
Then, if its mass variable $z_j[k+1]$ is nonzero, (i) it updates its state variables to be equal to the mass variables, and (ii) it chooses randomly an out-neighbor (or itself) and transmits the mass variables $y_j[k+1]$ and $z_j[k+1]$. 
Eventually, after a finite number of time steps, the ratio of state variables $y^s_j[k+1] / z^s_j[k+1]$ of each node $v_j$ is equal to the average degree in the network. 

\textbf{Iteration - Step~$2$. Distributed Stopping:}
Every $k = t D' +1$ time steps, where $t \in \mathbb{N}$, each node $v_j$ sets its voting variables $m_j$ and $M_j$ to be equal to the fraction $y^s_j[k] / z^s_j[k]$ of the state variables. 
It broadcasts its voting variables to its out-neighbors and receives the corresponding $m_i$ and $M_i$ from its in-neighbors $v_i\in\mathcal{N}_j^{-}$. 
It stores the $\min$ and $\max$ among all received and its own voting values to the variables $m_j$ and $M_j$, respectively. Then, the $\min-$ and $\max-$consensus algorithms are performed for $D'-1$ steps. 
When $k = (t+1) D'$ time steps, each node $v_j$ checks whether $M_j$, $m_j$ have equal values; if this holds, then every node terminates the operation of the algorithm. 
If not, the process continues, with step $k = (t+1) D' +1$, in which each node  $v_j$ sets its voting variables $m_j$ and $M_j$ to be equal to the fraction $y^s_j[k] / z^s_j[k]$ of the state variables, and the $\min-$ and $\max-$consensus algorithms are restarted. 
Note that the $m_j$ and $M_j$ are fractions of integers. 
Therefore, for both the $\min-$ and $\max-$consensus, each node $v_j$ at time step $k$ sends a pair of integer values $(y^s_j[k], z^s_j[k] )$ whose ratio $y^s_j[k]/z^s_j[k]$ is used during the consensus operation. 

% \todo{$\text{flag}_j$ is never used.. replace with true + explain how max,min are performed with pairs of integers}

\begin{varalgorithm}{1}
\caption{Average Degree Computation Algorithm}
\noindent \textbf{Input:} A strongly connected digraph $\mathcal{G}_d = (\mathcal{V}, \mathcal{E})$ with $n=|\mathcal{V}|$ nodes and $m=|\mathcal{E}|$ edges. 
Each node $v_j\in \mathcal{V}$ has knowledge of an upper bound $D'$ of the network diameter. \\
\textbf{Initialization:} Each node $v_j \in \mathcal{V}$: 
\begin{list4}
\item[$1)$] assigns a nonzero probability $b_{lj}$ to each of its outgoing edges $m_{lj}$, where $v_l \in \mathcal{N}^+_j \cup \{v_j\}$, as follows
\begin{align*}
b_{lj} = \left\{ \begin{array}{ll}
         \frac{1}{1 + \mathcal{D}_j^+}, & \mbox{if $l = j$ or $v_{l} \in \mathcal{N}_j^+$,}\\
         0, & \mbox{if $l \neq j$ and $v_{l} \notin \mathcal{N}_j^+$,}\end{array} \right. 
\end{align*} 
\item[$2)$] sets $y_j[1] := \mathcal{D}_j^+$, $z_j[1] = 1$. 
% and $\text{flag}_j = 0$. 
\item[$3)$] sets $y^s_j[1] := y_j[1]$, $z^s_j[1] = 1$, $q^s_j[1] := y^s_j[1] / z^s_j[1]$. 
\item[$4)$] chooses $v_l \in \mathcal{N}^+_j \cup \{v_j\}$ randomly according to $b_{lj}$, and transmits $y_j[1]$ and $z_j[1]$ towards $v_l$. 
\end{list4} 
\textbf{Iteration:} For $k=1,2,\dots$, each node $v_j \in \mathcal{V}$, does the following: 
\begin{list4} 
\item \textbf{while} 
% $\text{flag}_j = 0$ 
\textbf{\textit{true}}
\textbf{then} 
\begin{list4a}
\item[$1)$] \textbf{if} $k \mod D' = 1$ \textbf{then} sets $M_j = m_j = y^s_j[k] / z^s_j[k]$; 
\item[$2)$] broadcasts $M_j$, $m_j$ to every $v_{l} \in \mathcal{N}_j^+$; 
\item[$3)$] receives $M_i$, $m_i$ from every $v_{i} \in \mathcal{N}_j^-$; 
\item[$4)$] sets $M_j =\displaystyle \max_{v_{i} \in \mathcal{N}_j^-\cup \{ v_j \}} M_i$, $m_j =\displaystyle \min_{v_{i} \in \mathcal{N}_j^-\cup \{ v_j \}} m_i$; 
\item[$5)$] receives $y_i[k]$ and $z_i[k]$ from $v_i \in \mathcal{N}_j^-$ and updates according to \eqref{subeq:1a}-\eqref{subeq:1b};
%\begin{equation}\label{no_del_eq_y_no_oscil}
%y_j[k+1] = y_j[k] + \sum_{v_i \in \mathcal{N}_j^-}  w_{ji}[k] \ y_i[k] ,
%\end{equation} 
%\begin{equation}\label{no_del_eq_z_no_oscil}
%z_j[k+1] = z_j[k] + \sum_{v_i \in \mathcal{N}_j^-} w_{ji}[k] \ z_i[k] ,
%\end{equation}
%where $w_{ji}[k] = 1$ if node $v_j$ receives $c^{y}_{ji}[k]$, $c^{z}_{ji}[k]$ from $v_i \in \mathcal{N}_j^-$ at iteration $k$ (otherwise $w_{ji}[k] = 0$); 
\item[$6)$] \textbf{if} $z_j[k+1] > 1$, \textbf{then}
\begin{list4a}
\item[$6.1)$] sets $z^s_j[k+1] = z_j[k+1]$, $y^s_j[k+1] = y_j[k+1]$, 
$
q^s_j[k] = \frac{y^s_j[k]}{z^s_j[k]} \ ;
$
\item[$6.2)$] chooses $v_l \in \mathcal{N}^+_j \cup \{v_j\}$ randomly according to $b_{lj}$, and transmits $y_j[k+1]$ and $z_j[k+1]$ towards $v_l$. 
\end{list4a}
\item[$7)$] \textbf{if} $k \mod D' = 0$ \textbf{then}, \textbf{if} $M_j = m_j$ \textbf{then} stops operation. 
\end{list4a}
\end{list4}
\textbf{Output:} \eqref{average_degree_def} holds for every $v_j \in \mathcal{V}$. 
\label{algorithm_averagedegree}
\end{varalgorithm}

We now analyze the convergence of Algorithm~\ref{algorithm_averagedegree}. 
% Specifically, we show that during the operation of Algorithm~\ref{algorithm_averagedegree}, each node $v_j$ addresses problem $\textbf{P1.}$ presented in Section~\ref{sec:probForm} in a finite number of time steps with high probability. 

%\todo{??? For both theorems: The proof of the paper is not rigorous enough, and the key content is vague. The definition of convergence with probability needs to be clearly given, and then the corresponding conclusions should be strictly proved according to the definition. - THEMIS WILL FIX}

\begin{theorem}\label{thm:degree_estimation}
Consider a strongly connected digraph $\mathcal{G}_d = (\mathcal{V}, \mathcal{E})$ with $n=|\mathcal{V}|$ nodes and $m=|\mathcal{E}|$ edges. 
At time step $k=1$, each node $v_j$ follows the Initialization and Iteration steps as described in Algorithm~\ref{algorithm_averagedegree}. 
For any probability $p_0$, where $0 < p_0 < 1$, there exists $k_0 \in \mathbb{N}$, so that with probability at least $p_0$, every node $v_j$ addresses problem \textbf{P1} in Section~\ref{sec:probForm}. 
\end{theorem}

\begin{proof}
See Appendix~\ref{appendix:A}.
\end{proof}

% ===============================================
%
%
% ALGORITHM
%
%
% ===============================================
\section{Distributed Network Size Computation}
\label{sec:distr_alg_networksize}

In this section we present a distributed algorithm which solves problem $\textbf{P2}$. 
Our algorithm is detailed below as Algorithm~\ref{algorithm_networksize}. 
%Similarly to Algorithm~\ref{algorithm_averagedegree}, Algorithm~\ref{algorithm_networksize} also operates solely with quantized values (i.e., each node processes and transmits quantized values) and converges in finite time \AR{almost surely (i.e., with probability one)}. 
Note here that Assumption~\ref{Diam_known} also holds during the operation of the proposed algorithm and we  additionally  make the following assumption. 

\begin{assum}\label{leader_vote}
All nodes $v_j \in \mathcal{V}$ have knowledge of a constant $U_v \in \mathbb{N}$. 
\end{assum}

Assumption~\ref{leader_vote} is necessary for executing the $\max$-consensus algorithm in order to elect a leader node with high probability. This can be a preset value for running the protocol, irrespective of the network.

We now describe the main operations of Algorithm~\ref{algorithm_networksize}. 
% The initialization involves the following steps: 

\textbf{Initialization - Step~$1$. Probability Assignment:} 
Each node $v_j \in \mathcal{V}$ assigns a nonzero probability to each outgoing edge and a virtual self edge, so that the sum of nonzero probabilities is equal to one. 

\textbf{Initialization - Step~$2$. Leader Election:} 
Each node executes Algorithm~\ref{algorithm_LeaderElection}. 
During its operation, each node in the network executes a max-consensus for $U_v D'$ time steps (i.e., executes a max-consensus algorithm $U_v$ times). 
More specifically, each node randomly picks a quantized value and executes the first max-consensus for $D'$ time steps. 
Once the first max-consensus converges, we have that (i) the node (or nodes) that picked the maximum value pick again randomly a quantized value, and (ii) the node (or nodes) that did not pick the maximum value choose a value equal to $-1$. 
Then, the max-consensus is executed again for $D'$ time steps. 
This process is repeated $U_v$ times (i.e., for $U_v D'$  time steps). 
After the execution of Algorithm~\ref{algorithm_LeaderElection}, we have that one node $v_j \in \mathcal{V}$ is the leader (i.e., $\text{flag}^{\text{ld}}_j = 1$) and every node $v_i \in \mathcal{V} \setminus \{ v_j \}$ is a follower (i.e., $\text{flag}^{\text{ld}}_i = 0$), with high probability. 

\textbf{Initialization - Step~$3$. Initialization of Mass and State Variables:} 
Each node initializes its mass variables and its state variables according to the result of Algorithm~\ref{algorithm_LeaderElection} (i.e., the leader node initializes its mass variables different than the follower nodes). 
Specifically, the leader node initializes both $y$ and $z$ to be equal to $1$. 
Every follower node initializes $y$ to be equal to $0$ and $z$ to be equal to $1$. 
Then, every node sets its state variables to be equal to the mass variables.  

% The iteration involves the following steps: 

\textbf{Iteration - Step~$1$. Calculating the Network Size:} 
Every node $v_j$ receives the transmitted mass variables of its in-neighbors, and sums them with the stored mass variables. 
Then, if its mass variable $z_j[k+1]$ is nonzero, (i) it updates its state variables to be equal to the mass variables, and (ii) it chooses randomly an out-neighbor (or itself) and transmits the mass variables $y_j[k+1]$ and $z_j[k+1]$. 
Eventually, after a finite number of time steps, the state variable $z^s_j[k+1]$ of each node $v_j$ is equal to the number of nodes in the network. 

\textbf{Iteration - Step~$2$. Distributed Stopping:} 
The operation of this step is identical to ``Iteration - Step~$1$. Distributed Stopping'' of Algorithm~\ref{algorithm_averagedegree}. 
It is omitted due to space considerations. 
% For determining whether convergence has been achieved, every $k = 1 + \sigma D$ time steps, where $\sigma \in \mathbb{N}$, it sets its voting variables $m_j$ and $M_j$ to be equal to $y^s_j[k]$. 
% It broadcasts its voting variables to its out-neighbors and receives the broadcasted $m_i$ and $M_i$ from its in-neighbors. 
% It stores the min and max received values to the variables $m_j$ and $M_j$, respectively. 
% Every $k = \sigma D$ time steps, where $\sigma \in \mathbb{N}$, every node checks whether $M_j$, $m_j$ have equal values; if this holds, then every node stops the operation of the algorithm. 

\begin{varalgorithm}{2}
\caption{Distributed Network Size Computation Algorithm}
\noindent \textbf{Input:} A strongly connected digraph $\mathcal{G}_d = (\mathcal{V}, \mathcal{E})$ with $n=|\mathcal{V}|$ nodes and $m=|\mathcal{E}|$ edges. 
Each node $v_j\in \mathcal{V}$ has knowledge of an upper bound $D'$ of the network diameter. \\
\textbf{Initialization:} Each node $v_j \in \mathcal{V}$ does the following: 
\begin{list4}
\item[$1)$] Assigns a nonzero probability $b_{lj}$ to each of its outgoing edges $m_{lj}$, where $v_l \in \mathcal{N}^+_j \cup \{v_j\}$, as follows
\begin{align*}
b_{lj} = \left\{ \begin{array}{ll}
         \frac{1}{1 + \mathcal{D}_j^+}, & \mbox{if $l = j$ or $v_{l} \in \mathcal{N}_j^+$,}\\
         0, & \mbox{if $l \neq j$ and $v_{l} \notin \mathcal{N}_j^+$.}\end{array} \right. 
\end{align*} 
\item[$2)$] Calls Algorithm~\ref{algorithm_LeaderElection}; 
\begin{list4a}
\item[$2.1)$] \textbf{if} $\text{flag}^{\text{ld}}_j = 1$, sets $y_j[1] := 1$, $z_j[1] = 1$; 
\item[$2.2)$] \textbf{if} $\text{flag}^{\text{ld}}_j = 0$, sets $y_j[1] := 0$, $z_j[1] = 1$; 
\end{list4a}
% \item[$3)$] Sets $\text{flag}_j = 0$. 
\item[$3)$] Sets $y^s_j[1] := y_j[1]$, $z^s_j[1] = 1$. 
\end{list4} 
\textbf{Iteration:} For $k=1,2, \dots$, each node $v_j \in \mathcal{V}$, does the following: 
\begin{list4} 
\item \textbf{while} 
% $\text{flag}_j = 0$ 
\textbf{\textit{true}}
\textbf{then} 
\begin{list4a}
\item[$1)$] \textbf{if} $k \mod D' = 1$ \textbf{then} sets $M_j = m_j = y^s_j[k]$; 
\item[$2)$] steps $2 - 7$ are same as Algorithm~\ref{algorithm_averagedegree}; 
% \item[$2)$] broadcasts $M_j$, $m_j$ to every $v_{l} \in \mathcal{N}_j^+$; 
% \item[$3)$] receives $M_i$, $m_i$ from every $v_{i} \in \mathcal{N}_j^-$; 
% \item[$4)$] sets $M_j = \max_{v_{i} \in \mathcal{N}_j^-\cup \{ v_j \}} \ M_i$, $m_j = \min_{v_{i} \in \mathcal{N}_j^-\cup \{ v_j \}} \ m_i$; 
% \item[$5)$] receives $y_i[k]$ and $z_i[k]$ from $v_i \in \mathcal{N}_j^-$ and sets 
% \begin{equation}\label{no_del_eq_y_no_oscil}
% y_j[k+1] = y_j[k] + \sum_{v_i \in \mathcal{N}_j^-}  w_{ji}[k] \ y_i[k] ,
% \end{equation} 
% \begin{equation}\label{no_del_eq_z_no_oscil}
% z_j[k+1] = z_j[k] + \sum_{v_i \in \mathcal{N}_j^-} w_{ji}[k] \ z_i[k] ,
% \end{equation}
% where $w_{ji}[k] = 1$ if node $v_j$ receives $c^{y}_{ji}[k]$, $c^{z}_{ji}[k]$ from $v_i \in \mathcal{N}_j^-$ at iteration $k$ (otherwise $w_{ji}[k] = 0$); 
% \item[$6)$] \textbf{if} $z_j[k+1] > 1$, \textbf{then}
% \begin{list4a}
% \item[$6.1)$] sets $z^s_j[k+1] = z_j[k+1]$, $y^s_j[k+1] = y_j[k+1]$; 
% \item[$6.2)$] chooses $v_l \in \mathcal{N}^+_j \cup \{v_j\}$ randomly according to $b_{lj}$, and transmits $y_j[k+1]$ and $z_j[k+1]$ towards $v_l$. 
% \end{list4a}
% \item[$7)$] \textbf{if} $k \mod D = 0$ \textbf{then}, \textbf{if} $M_j = m_j$ \textbf{then} stops operation. 
\end{list4a}
\end{list4}
\textbf{Output:} \eqref{number_nodes_network} holds for every $v_j \in \mathcal{V}$. 
\label{algorithm_networksize}
\end{varalgorithm}

% \todo{
% $M_j'$ - leader voting variable 
% $\eta_j'$ - initial value in the voting leader 
% $Up_j$ - is upper bound in the set to which the values are chosen to be the leader. 
% $\text{flag}^{\text{ld}}_j = 1$, initiates at 1, and becomes 0 once the node understands is not leader. 
% if $\text{flag}^{\text{ld}}_j$ becomes zero, node vj removes one from y value. 
% }

\begin{varalgorithm}{2A}
\caption{Leader Election with Quantized Information}
\noindent \textbf{Input:} A strongly connected digraph $\mathcal{G}_d = (\mathcal{V}, \mathcal{E})$ with $n=|\mathcal{V}|$ nodes and $m=|\mathcal{E}|$ edges. 
Each node $v_j\in \mathcal{V}$ has knowledge of an upper bound $D'$ of the network diameter. \\
\textbf{Initialization:} Each node $v_j \in \mathcal{V}$ sets $\text{flag}^{\text{ld}}_j = 1$. \\ 
% \begin{list4} 
% \item[$1)$] Sets $\text{flag}^{\text{ld}}_j = 1$. 
% \end{list4} 
\textbf{Iteration:} For $k = 1,2, \dots, U_v D'$, each node $v_j \in \mathcal{V}$, does the following: 
\begin{list4} 
\item[$1)$] \textbf{if} $k \mod D' = 1$ \textbf{then} 
\begin{list4a}
\item[$1.1)$] if $\text{flag}^{\text{ld}}_j = 1$ \textbf{then} chooses randomly $\eta_j' \in \{ 0, 1, ..., \ Up_j \}$, and sets $M_j' = \eta_j'$; 
\item[$1.2)$] if $\text{flag}^{\text{ld}}_j = 0$ \textbf{then} sets $\eta_j' = -1$ and $M_j' = \eta_j'$; 
\end{list4a} 
\item[$2)$] broadcasts $M_j'$ to every $v_{l} \in \mathcal{N}_j^+$; 
\item[$3)$] receives $M_i'$ from every $v_{i} \in \mathcal{N}_j^-$; 
\item[$4)$] sets $M_j' = \max_{v_{i} \in \mathcal{N}_j^-\cup \{ v_j \}} \ M_i'$; 
\item[$5)$] \textbf{if} $k \mod D' = 0$ \textbf{then} \textbf{if} $M_j' \neq \eta_j'$ \textbf{then} sets $\text{flag}^{\text{ld}}_j = 0$; 
\end{list4}
\textbf{Output:} $\text{flag}^{\text{ld}}_j$ for every $v_j \in \mathcal{V}$. 
\label{algorithm_LeaderElection}
\end{varalgorithm}

%\textbf{Parallel Execution of Algorithm~\ref{algorithm_LeaderElection}, and Algorithm~\ref{algorithm_networksize}.}

We now analyze the convergence of Algorithm~\ref{algorithm_networksize}. 

\begin{theorem}\label{thm:size_estimation}
Consider a strongly connected digraph $\mathcal{G}_d = (\mathcal{V}, \mathcal{E})$ with $n=|\mathcal{V}|$ nodes and $m=|\mathcal{E}|$ edges. 
At time step $k=1$, each node $v_j$ follows the Initialization and Iteration steps as described in Algorithm~\ref{algorithm_networksize}. 
For any probability $p_0$, where $0 < p_0 < 1$, there exists $k_0 \in \mathbb{N}$, so that with probability at least $p_0$, every node $v_j$ addresses problem \textbf{P2} in Section~\ref{sec:probForm}. 
\end{theorem}

\begin{proof}
%The proof of Theorem~\ref{thm:main} is presented in Appendix~\ref{appendix:A}.
See Appendix~\ref{appendix:B}.
\end{proof}

\begin{remark} 
Apart from guaranteeing that the exact number of nodes is computed, Algorithm~\ref{algorithm_networksize} establishes also that the process is terminated (so that other algorithms can be initiated after the convergence of Algorithm~\ref{algorithm_networksize}). 
If operation termination was not required, then a variation of Algorithm~\ref{algorithm_networksize} could be used to calculate the network size in a finite number of steps without electing a leader node. 
Specifically, during this variation each node $v_j$ initializes $z_j[1] = 1$ and executes iteration steps $2$--$6$ of Algorithm~\ref{algorithm_networksize}. 
During this execution, there exists $k_0$, for which $z^s_j[k] = n$ for every $v_j \in \mathcal{V}$, for $k \geq k_0$. 
\end{remark}
% \textbf{Comparison with Previous Literature.} 
% Algorithm~\ref{algorithm_averagedegree} and Algorithm~\ref{algorithm_networksize} are a major departure from the current literature. 
% Specifically, in the current literature every algorithm operates with real values. 
% This increases significantly the 

% \begin{remark}
% \todo{speak that the result is calculated exactly. Also, emphasize on the contribution: the first leader election algorithm which operates with quantized values and converges according to a probability. 
% (i) the distributed stopping which is adjusted to the exact quant consensus + the algorithm has asynchronous updates + Speak that results can be extended for undirected graphs}. 
% \end{remark}

% \todo{REMARK: build on: Algorithm of Iman does not give the exact value - for large scale networks Iman does not have the exact value - therefore we cannot compute the exact number of nodes. 
% Also if Iman has large n, -- 1/n is very small number, then small errors in small values will lead to great miscalculation of exact value. 
% }

\begin{remark}    
%The operation of Algorithm~\ref{algorithm_networksize} relies on the result of Algorithm~\ref{algorithm_LeaderElection}. More specifically, 
Algorithm~\ref{algorithm_LeaderElection} is executed as an initialization step of Algorithm~\ref{algorithm_networksize}. 
%This means that the Iteration Steps of Algorithm~\ref{algorithm_networksize} start once Algorithm~\ref{algorithm_LeaderElection} terminates. 
%This may increase significantly the required number of time steps for convergence of Algorithm~\ref{algorithm_networksize} (see Fig.~\ref{parallel_consec} (A)). 
However, %it is important to note that 
both can be executed in parallel. 
%More specifically, the Iteration Steps of Algorithm~\ref{algorithm_LeaderElection} are executed for $U_v D'$ time steps along with the Iteration Steps of Algorithm~\ref{algorithm_networksize}. 
The strategy of parallel execution significantly decreases the required number of time steps for convergence of Algorithm~\ref{algorithm_networksize} as demonstrated in Section~\ref{sec:results} (see Fig.~\ref{parallel_consec}). 
\end{remark}

% ===============================================
%
%
% SIMULATIONS
%
%
% ===============================================

\vspace{-.5cm}

\section{Simulation Results} \label{sec:results}

%\todo{??? numerical comparison with the existing approach should be presented - OK ???}

% \todo{OTHER ALGORITHMS ARE EITHER ASYMPTOTIC - WE DO NOT COMPARE BECAUSE OTHER ALGORITHMS ARE NOT QUANTIZED AND CALCULATE RESULT WITH ERROR ++ MENTION THIS ALSO TO INTRODUCTION}

In this section, we present simulation results in order to demonstrate the operation of our proposed algorithms and their potential advantages. 
For both algorithms we focus on a random digraph of $20$ nodes and show how the nodes' states converge to the desired value in finite time. %\AR{almost surely (i.e., with probability one)}. 
We also demonstrate the distributed stopping capabilities of our algorithms. 
Finally, we analyze numerically the convergence time of our algorithms, and we present the distributions of the required number of time steps for convergence. 
Our simulations emphasize the novelty of our algorithms which, to the best of our knowledge, are the first that use quantized values to calculate the exact values of the average degree and size of a network while also providing strong theoretical guarantees.
Note here that comparing against other algorithms from current literature is not feasible  due to the fact that all other methods in the literature operate with real-values and exhibit asymptotic convergence. 
%Specifically, current approaches require an infinite number of time steps to calculate the exact values of the parameters which we calculate in finite time. 
%However, an extended analysis of the operational advantages of our algorithms in terms of communication efficiency and convergence speed compared to algorithms in the literature will be one in an extended version of this paper.}

%due to space limitations, we do not provide 

\vspace{.3cm}

\noindent
\textbf{Average Degree Computation of a Random Network of $20$ Nodes.} 
In this section we demonstrate the operation of Algorithm~\ref{algorithm_averagedegree} over a random digraph of $20$ nodes. 
Each edge of the digraph was generated with probability $0.5$, and the diameter is equal to $D=3$. 
The average degree of the digraph is equal to $\dfrac{206}{20} = 10.3$. 
During the execution of Algorithm~\ref{algorithm_averagedegree} each node has knowledge of an upper bound on the network's diameter equal to $D' = 4$. 

In Fig.~\ref{10_3}, we plot the evolution of the state variable $q_j^s[k]$ of every node $v_j$. 
We can see that Algorithm~\ref{algorithm_averagedegree} converges after $74$ time steps to the exact solution. 
Specifically, each node calculates the quantized fraction ${106}/{20}$ which is equal to the average degree in the network. 
Furthermore, we can see that after $78$ iterations each node terminates its operation since it has knowledge of $D' = 4$ which is an upper bound on the network diameter. 

\begin{figure}[t]
    \centering
    \includegraphics[width=.82\linewidth]{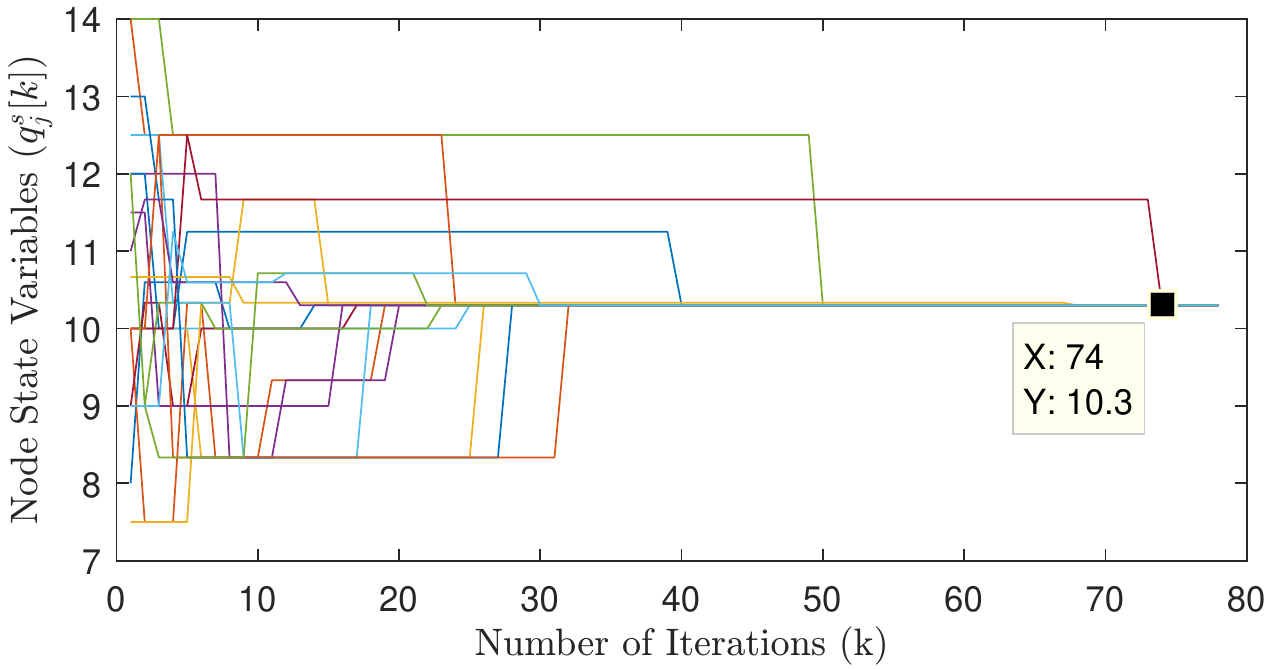}
    \caption{Execution of Algorithm~\ref{algorithm_averagedegree} over a random digraph of $20$ nodes with diameter $D=3$.}
    \label{10_3}
\end{figure}

In Fig.~\ref{120_96_10000try}, we present $10000$ executions of Algorithm~\ref{algorithm_averagedegree} over a random digraph of $20$ nodes. 
Each edge of the digraph was generated with probability $0.5$, and the diameter is equal to $D=3$ for each of the $10000$ executions. 
The average number of time steps for convergence of Algorithm~\ref{algorithm_averagedegree} is equal to $120.96$. 
In Fig.~\ref{120_96_10000try} we can see that in most cases Algorithm~\ref{algorithm_averagedegree} requires $70$--$160$ iterations for convergence. 

\begin{figure}[t]
    \centering
    \includegraphics[width=.85\linewidth]{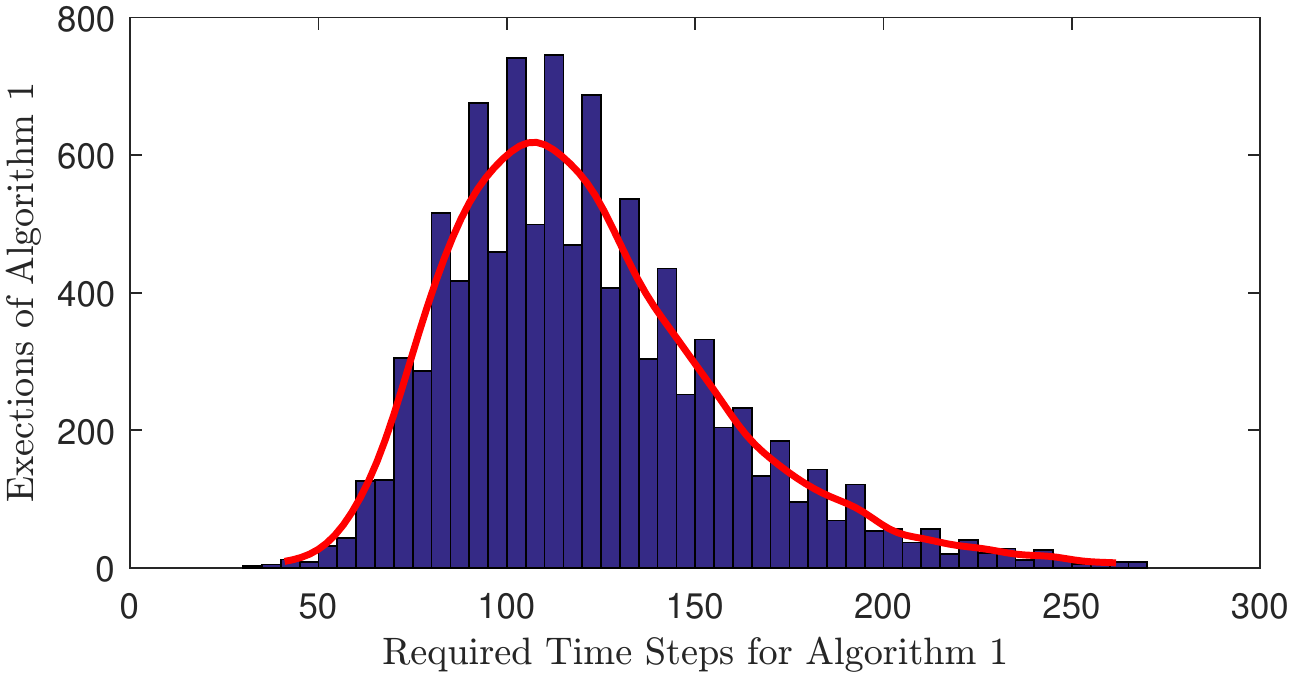}
    \caption{Required time steps for convergence of $10000$ executions of  Algorithm~\ref{algorithm_averagedegree} over a random digraph of $20$ nodes with diameter $D=3$.}
    \label{120_96_10000try}
\end{figure}

\vspace{.3cm}

\noindent
\textbf{Size Computation of a Random Network of $20$ Nodes.} 
In this section we demonstrate the operation of Algorithm~\ref{algorithm_networksize} over a random digraph of $20$ nodes. 
The parameters of the digraph are the same as in the previous example where we demonstrated Algorithm~\ref{algorithm_averagedegree}, and each node also has knowledge of $D' = 4$. 
Furthermore, we have that $U_v = 20$. 
This means that Algorithm~\ref{algorithm_LeaderElection} is executed for $80$ time steps. 

In Fig.~\ref{20nodes}, we plot the evolution of the state variable $z_j^s[k]$ of every node $v_j$. 
We can see that Algorithm~\ref{algorithm_networksize} calculates after $160$ time steps the quantized fraction ${1}/{20}$, where the denominator is equal to the number of nodes in the network. 
Finally, after $164$ iterations each node terminates its operation since it has knowledge of $D' = 4$. 
Note that in Fig.~\ref{20nodes}, we plot the evolution of $z_j^s[k]$ for time steps $k \geq 80$, since Algorithm~\ref{algorithm_LeaderElection} is executed for $80$ time steps. 

\begin{figure}[t]
    \centering
    \includegraphics[width=.85\linewidth]{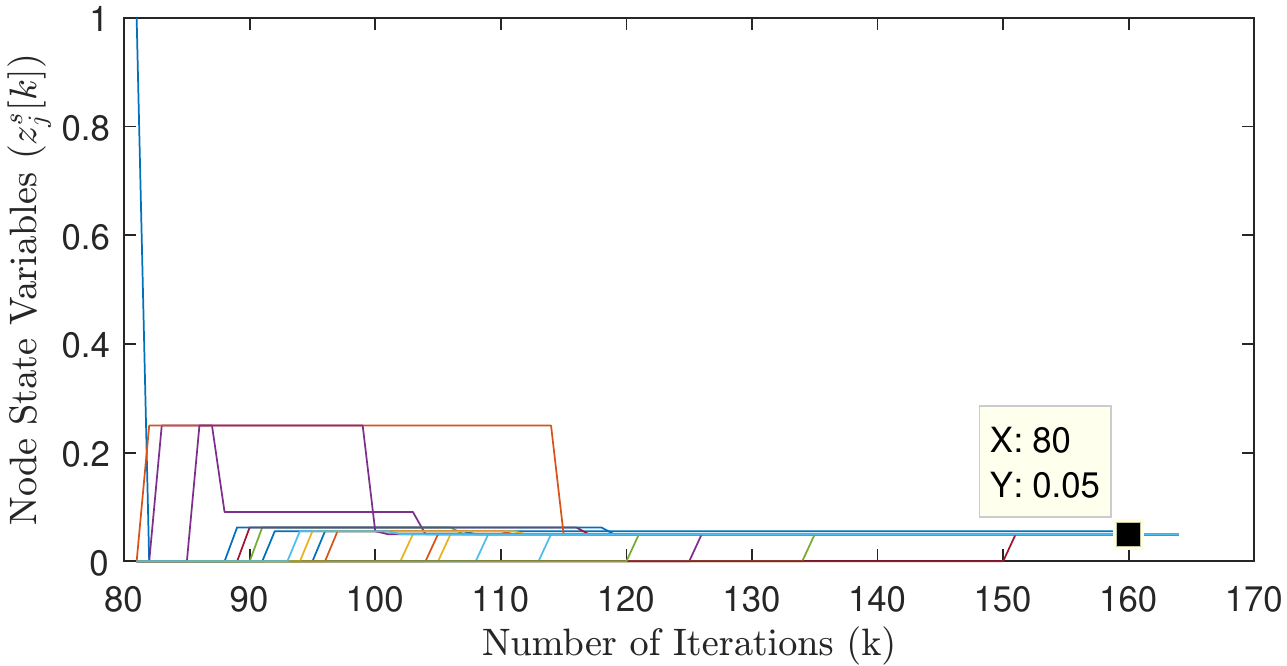}
    \caption{Execution of Algorithm~\ref{algorithm_networksize} over a random digraph of $20$ nodes with diameter $D=3$.} 
    \label{20nodes}
\end{figure}

In Fig.~\ref{parallel_consec}, we present $10000$ executions of Algorithm~\ref{algorithm_networksize} over a random digraph of $20$ nodes. 
Each edge of the digraph was generated with probability $0.5$, the diameter is equal to $D=3$ for each of the $10000$ executions, and each node also has knowledge of $D' = 4$. 
In Fig.~\ref{parallel_consec} (A) we execute Algorithm~\ref{algorithm_LeaderElection} before Algorithm~\ref{algorithm_networksize} for $U_v = 20$ and $D' = 4$ (i.e., we execute Algorithm~\ref{algorithm_LeaderElection} for $80$ time steps). 
In Fig.~\ref{parallel_consec} (B) we execute Algorithm~\ref{algorithm_LeaderElection} in parallel with Algorithm~\ref{algorithm_networksize}. 
The average number of time steps for convergence of Algorithm~\ref{algorithm_networksize} is equal to $198.64$ for (A), and $119.60$ for (B). 
In Fig.~\ref{parallel_consec} (A) we can see that in most cases Algorithm~\ref{algorithm_networksize} requires $150$--$240$ iterations for convergence. 
However, in Fig.~\ref{parallel_consec}~(B) we can see that in most cases it requires $80$--$140$. 
More specifically, in Fig.~\ref{parallel_consec}~(B) we can see that $1400$ executions of  Algorithm~\ref{algorithm_networksize} require $80$ -- $90$ time steps to converge. 
This is mainly due to the fact that the Iteration Steps of Algorithm~\ref{algorithm_networksize} have reached convergence, but nodes need to implement Algorithm~\ref{algorithm_LeaderElection} for $80$ time steps. 
As a result, parallel execution of Algorithm~\ref{algorithm_LeaderElection} with Algorithm~\ref{algorithm_networksize} is better, since a significantly smaller number of time steps is required for convergence. 

\begin{figure}[t]
    \centering
    \includegraphics[width=.49\linewidth]{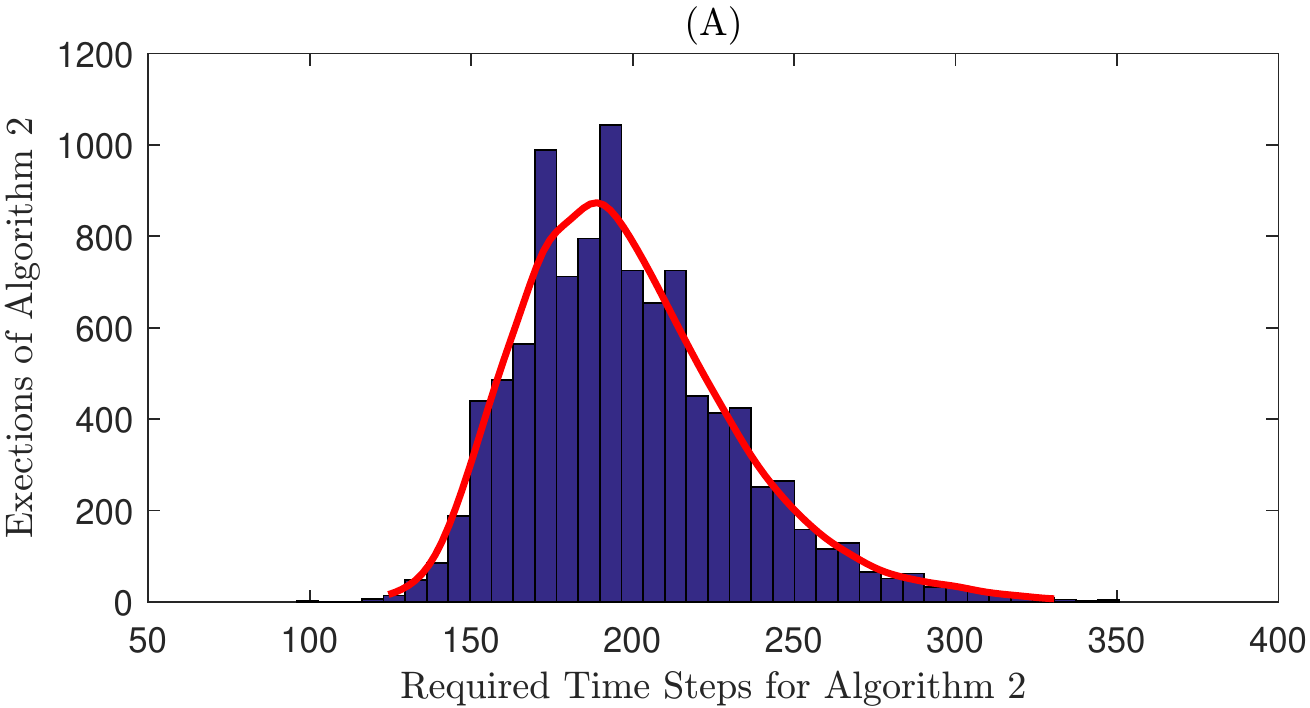} %\\ \vspace{.3cm}
    \includegraphics[width=.49\linewidth]{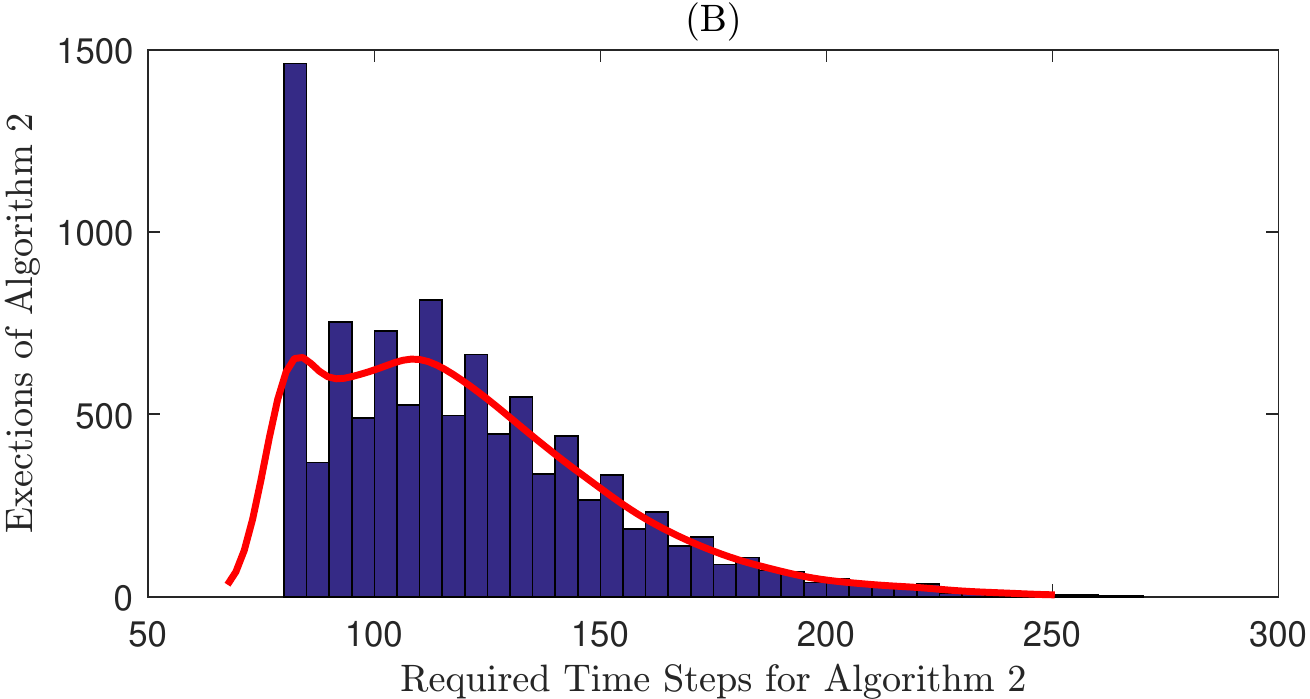}
    \caption{Required time steps for convergence of $10000$ executions of  Algorithm~\ref{algorithm_networksize} over a random digraph of $20$ nodes with diameter $D=3$. (A) Execution of Algorithm~\ref{algorithm_LeaderElection} before Algorithm~\ref{algorithm_networksize}. (B) Execution of Algorithm~\ref{algorithm_LeaderElection} in parallel with Algorithm~\ref{algorithm_networksize}.} 
    \label{parallel_consec}
\end{figure}

\textbf{Leader Election via Algorithm~\ref{algorithm_LeaderElection}.} 
We now analyze the probability of having one leader after the execution of Algorithm~\ref{algorithm_LeaderElection}, over a random digraph of $20$ nodes. 
In Fig.~\ref{up_4_5_8_10000iter} we present $10000$ executions of Algorithm~\ref{algorithm_LeaderElection} for $U_v = 10$. 
We plot the instances for which we have more than one leader nodes during the execution of Algorithm~\ref{algorithm_LeaderElection}, for the cases where (i) $\eta_j' \in \{ 0, 1, ..., \ 15 \}$, (ii) $\eta_j' \in \{ 0, 1, ..., \ 31 \}$, and (iii) $\eta_j' \in \{ 0, 1, ..., \ 255 \}$ for every node $v_j$. 
We can see that if we increase the range of $\eta_j'$ for every node $v_j$, the convergence rate of Algorithm~\ref{algorithm_LeaderElection} improves greatly. 
Also, after $5$ executions of Algorithm~\ref{algorithm_LeaderElection} (i.e., $U_v = 5$), the number of instances where we have multiple leader nodes is equal to zero. 
This means that during Algorithm~\ref{algorithm_LeaderElection} for $U_v \geq 5$, only one node is leader node. 

\begin{figure}[t]
    \centering
    \includegraphics[width=\linewidth]{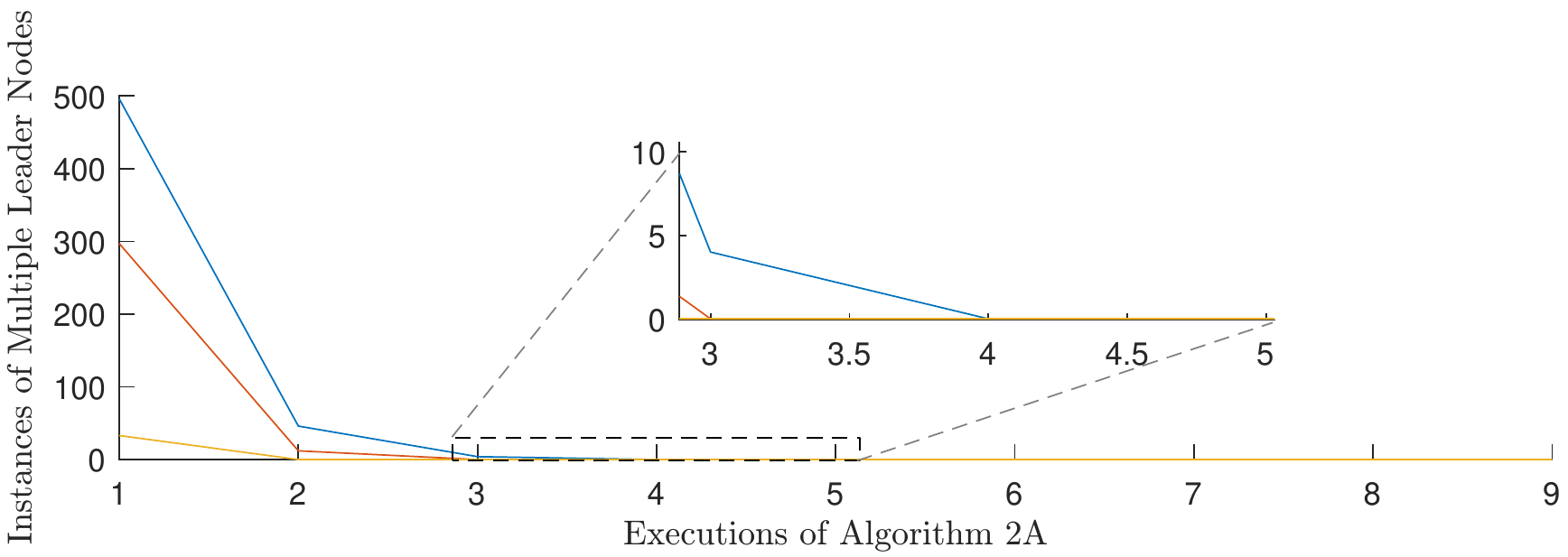}
    \caption{Instances where multiple nodes are leader nodes during $10000$ executions of Algorithm~\ref{algorithm_LeaderElection}.} 
    \label{up_4_5_8_10000iter}
\end{figure}

% \todo{insert table of probabilities for Algorithm2A.. probability to have one leader... prob to have multiple leaders}. 

% \todo{table probability of one leader, >1 leader, vs UP 1 2 5 10, 20, scale 0-255} 

% ===============================================
%
%
% CONCLUSIONS
%
%
% ===============================================

\vspace{-.4cm}

\section{Conclusions and Future Directions}\label{sec:conclusions}

%In this paper, we considered the problem of computing (i) the average degree, and (ii) the size of a given network. The operation of our algorithms rely on event-triggered updates and each node processes and transmits quantized values over directed networks. 
We showed that our algorithms are able to compute the exact values of the average degree and the size of a network after a finite number of time steps. 
Additionally, our algorithms allow each node to determine in a distributed fashion whether convergence has been achieved, and thus terminate its operation. 
Furthermore, in order to implement our algorithms, we present the first leader election strategy which relies on quantized operation (i.e., nodes process and transmit quantized information). 
Finally, we have demonstrated the operation of our algorithm over random directed networks and illustrated its finite time convergence. 

Extending our algorithms to achieve fully asynchronous operation which is desirable in large-scale networks, and operating over unreliable networks (e.g., packet dropping links) are two of our main future directions. 

\vspace{-.3cm}

% ------------------------------------------------------------------------------
% Bibliography
% ------------------------------------------------------------------------------
%\bibliographystyle{IEEEtran}
%\bibliography{bibliography}

%\vspace{.5cm}
\appendices
% ===============================================
%
%
% Proof of ....
%
%
% ===============================================
\section{Proof of Theorem~\ref{thm:degree_estimation}}
\label{appendix:A}

We first consider Lemma~\ref{Lemma1_prob}, which is necessary for our subsequent development. 

\begin{lemma}[\hspace{-0.00001cm}\cite{2020:RikosHadj_IFAC}]
\label{Lemma1_prob}
Consider a strongly connected digraph $\mathcal{G}_d = (\mathcal{V}, \mathcal{E})$ with $n=|\mathcal{V}|$ nodes and $m=|\mathcal{E}|$ edges.
Suppose that each node $v_j$ assigns a nonzero probability $b_{lj}$ to each of its outgoing edges $m_{lj}$, where $v_l \in \mathcal{N}^+_j \cup \{v_j\}$, as 
follows  
\begin{align*}
b_{lj} = \left\{ \begin{array}{ll}
         \frac{1}{1 + \mathcal{D}_j^+}, & \mbox{if $l = j$ or $v_{l} \in \mathcal{N}_j^+$,}\\
         0, & \mbox{if $l \neq j$ and $v_{l} \notin \mathcal{N}_j^+$.}\end{array} \right. 
\end{align*}
At time step $k=0$, node $v_j$ holds a ``token" while the other nodes $v_l \in \mathcal{V} - \{ v_j \}$ do not. 
Each node $v_j$ transmits the ``token'' (if it has it, otherwise it performs no transmission) according to the nonzero probability $b_{lj}$ it assigned to its outgoing edges $m_{lj}$. 
The probability $P^{D}_{T_i}$ that the token is at node $v_i$ after $D$ time steps satisfies 
$
P^{D}_{T_i} \geq (1+\mathcal{D}^+_{\max})^{-D} > 0 , 
$
where $\mathcal{D}^+_{\max} = \max_{v_j \in \mathcal{V}} \mathcal{D}^+_{j}$. 
\end{lemma}

We now consider Lemma~\ref{tokenmeet}, which analyzes the probability according to which two tokens performing a random walk visit a specific node at the same time step. 
The proof is similar to Lemma~\ref{Lemma1_prob},  \textit{mutatis mutandis}, and is omitted due to space limitations. 

\begin{lemma}
\label{tokenmeet}
Consider a strongly connected digraph $\mathcal{G}_d = (\mathcal{V}, \mathcal{E})$ with $n=|\mathcal{V}|$ nodes and $m=|\mathcal{E}|$ edges.
Suppose that each node $v_j$ assigns a nonzero probability $b_{lj}$ to each of its outgoing edges $m_{lj}$, where $v_l \in \mathcal{N}^+_j \cup \{v_j\}$, as follows  
\begin{align*}
b_{lj} = \left\{ \begin{array}{ll}
         \frac{1}{1 + \mathcal{D}_j^+}, & \mbox{if $l = j$ or $v_{l} \in \mathcal{N}_j^+$,}\\
         0, & \mbox{if $l \neq j$ and $v_{l} \notin \mathcal{N}_j^+$.}\end{array} \right. 
\end{align*}
At time step $k=0$, nodes $v_i$, $v_j$ hold a ``token" while the other nodes $v_l \in \mathcal{V} - \{ v_j, v_i \}$ do not (i.e., there are two tokens in the network). 
Each node $v_j$ transmits the ``token'' (if it has it, otherwise it performs no transmission) according to the nonzero probability $b_{lj}$ it assigned to its outgoing edges $m_{lj}$. 
After $D$ time steps, the probability $P^{D}_{TT_l}$ that the two tokens visit a specific node at the same time step is 
$
P^{D}_{TT_l} \geq \sum_{l=1}^{n} (1+\mathcal{D}^+_{\max})^{-(2D)} > 0 , 
$
where $\mathcal{D}^+_{\max} = \max_{v_j \in \mathcal{V}} \mathcal{D}^+_{j}$. 
\end{lemma}

Now we present the proof of Theorem~\ref{thm:degree_estimation}. 
%
% \begin{theorem}\label{Theorem2_Alg2_Converge}
% Consider a strongly connected digraph $\mathcal{G}_d = (\mathcal{V}, \mathcal{E})$ with $n=|\mathcal{V}|$ nodes and $m=|\mathcal{E}|$ edges. 
% At time step $k=1$, each node $v_j$ follows the Initialization and Iteration steps as described in Algorithm~\ref{algorithm1}. 
% For any probability $p_0$, where $0 < p_0 < 1$, there exists $k_0 \in \mathbb{N}$, so that with probability at least $p_0$, every node $v_j$ addresses problem \textbf{P1.} in Section~\ref{sec:probForm}. 
% \end{theorem}
%
%\begin{proof}
%The intuition of this proof is the following. 
The operation of Algorithm~\ref{algorithm_averagedegree} can be interpreted as the ``random walk'' of $n$ ``tokens'' in a Markov chain. 
Each token has a pair of values $y$, $z$. 
If two (or more) tokens visit the same node at the same time step $k$, they ``merge'' to a new single token (i.e., if two (or more) tokens merge, their $y$ values sum to the new $y$ value, and their $z$ values sum to the new $z$ value). 
Then, the new token performs a random walk in a Markov chain. 
Once all $n$ tokens merge to a final single token, this final single token has a pair of values $y$, $z$ whose ratio $y / z$ is equal to the average degree in the network. 
Thus, executing Algorithm~\ref{algorithm_averagedegree} for an additional finite number of time steps, this final single token will visit every node in the network. 

The structure of the proof comprises of three parts. 
In the fist part (\textbf{Part I}), we calculate the number of time steps $k_0'$ after which all tokens have merged to one single final token with probability at least $\sqrt{p_0}$. 
In the second part (\textbf{Part II}), we calculate the number of time steps $k_0''$ after which the final single token has visited every node in the network with probability at least $\sqrt{p_0}$. 
In the third part (\textbf{Part III}), we calculate the number of time steps $k_0$ after which \eqref{average_degree_def} holds for every node, and each node ceases transmissions. 

\textbf{Part I.} During the operation of Algorithm~\ref{algorithm_averagedegree}, from Lemma~\ref{tokenmeet} we have that, after $D$ time steps, the probability $P^{D}_{TT_l}$ that two (or more) tokens visit a specific node at the same time step is 
$
P^{D}_{TT_l} \geq \sum_{l=1}^{n} (1+\mathcal{D}^+_{\max})^{-(2D)} > 0 , 
$
where $\mathcal{D}^+_{\max} = \max_{v_j \in \mathcal{V}} \mathcal{D}^+_{j}$. 
This means that, after $D$ time steps, the probability $P^{D}_{NTT_l}$ that two (or more) tokens \textit{do not} visit a specific node at the same time step is 
\begin{equation}\label{two_not_visit}
    P^{D}_{NTT_l} \leq 1- \sum_{l=1}^{n} (1+\mathcal{D}^+_{\max})^{-(2D)} . 
\end{equation}
By extending this analysis, we choose $\varepsilon'$ (where $0 < \varepsilon' < 1$) for which it holds that
\begin{equation}\label{epsilon_value} 
    \varepsilon' \leq 1 - 2^{\frac{\log_2 \sqrt{p_0}}{n-1}} . 
\end{equation}
After $\tau' D$ time steps where 
\begin{equation}\label{tau_value} 
\tau' \geq \Big \lceil \dfrac{\log{\varepsilon'}}{\log{(1 - \sum_{l=1}^{n} (1+\mathcal{D}^+_{\max})^{-(2D)})}} \Big \rceil ,
\end{equation}
and $\varepsilon'$ fulfills \eqref{epsilon_value}, we have that the probability $P^{\tau D}_{NTT_l}$ that two (or more) tokens \textit{do not} visit a specific node at the same time step is 
% \begin{equation}\label{two_not_visit_tau} 
$
    P^{\tau D}_{NTT_l} \leq [P^{D}_{NTT_l}]^{\tau'} \leq \varepsilon' . 
$
% \end{equation}
This means that after $\tau' D$ time steps, the probability $P^{\tau D}_{TT_l}$ that two (or more) tokens visit a specific node at the same time step is $P^{\tau D}_{TT_l}  \geq 1 - \varepsilon'$.
Therefore, after $k_0' \geq (n-1) \tau' D$ time steps, the probability $P^{(n-1) \tau D}_{TT_l}$ that two (or more) tokens visit a specific node for $n-1$ instances at the same time step is 
\begin{equation}\label{all_visit_tau} 
    P^{(n-1) \tau D}_{TT_l}  \geq (1 - \varepsilon')^{(n-1)} \geq \sqrt{p_0} . 
\end{equation}

\textbf{Part II.} During the operation of Algorithm~\ref{algorithm_averagedegree}, from Lemma~\ref{Lemma1_prob} we have that the probability $P^{D}_{T_l}$ that a token visits a specific node after $D$ time steps is 
$P^{D}_{T_l} \geq (1+\mathcal{D}^+_{\max})^{-D} > 0$, where $\mathcal{D}^+_{\max} = \max_{v_j \in \mathcal{V}} \mathcal{D}^+_{j}$. 
This means that the probability $P^{D}_{NT_l}$ a token \textit{does not} visit a specific node after $D$ steps is 
\begin{equation}\label{one_not_visit}
    P^{D}_{NT_l} \leq 1- (1+\mathcal{D}^+_{\max})^{-D} . 
\end{equation}
We choose $\varepsilon''$ (where $0 < \varepsilon'' < 1$) for which it holds that 
\begin{equation}\label{one_epsilon_value} 
    \varepsilon'' \leq 1 - 2^{\frac{\log_2 \sqrt{p_0}}{n-1}} . 
\end{equation}
After $\tau'' D$ time steps where
\begin{equation}\label{one_tau_value} 
\tau'' \geq \Big \lceil \dfrac{\log{\varepsilon''}}{\log{(1 - (1+\mathcal{D}^+_{\max})^{-D})}} \Big \rceil ,
\end{equation}
and $\varepsilon''$ fulfills \eqref{one_epsilon_value}, we have that the probability $P^{\tau D}_{NT_l}$ that one token \textit{does not} visit a specific node is $P^{\tau D}_{NT_l} \leq [P^{D}_{NTT_l}]^{\tau''} \leq \varepsilon''$.
%\begin{equation}\label{one_not_visit_tau} 
%    P^{\tau D}_{NT_l} \leq [P^{D}_{NTT_l}]^{\tau''} \leq \varepsilon'' . 
%\end{equation}
This means that after $\tau'' D$ time steps, the probability $P^{\tau D}_{T_l}$ that one token visits a specific node is 
$
    P^{\tau D}_{TT_l}  \geq 1 - \varepsilon'' . 
$
Therefore, after $k_0'' \geq (n-1) \tau'' D$ time steps, the probability $P^{(n-1) \tau D}_{T_l}$ that one token visits a specific node for $n-1$ instances (i.e., it visits every node in the network) is
% \begin{equation}\label{one_all_visit_tau} 
$
    P^{(n-1) \tau D}_{T_l}  \geq (1 - \varepsilon'')^{(n-1)} \geq \sqrt{p_0} . 
$
% \end{equation}

\textbf{Part III.} 
During the operation of Algorithm~\ref{algorithm_averagedegree}, from Lemmas~\ref{Lemma1_prob} and~\ref{tokenmeet}, we can state that after $(n-1) \tau' D + (n-1) \tau'' D$ time steps, where $\tau'$ fulfills \eqref{tau_value} and $\tau''$ fulfills \eqref{one_tau_value}, we have that \eqref{average_degree_def} holds for every $v_j \in \mathcal{V}$ with probability at least $p_0$. 
Then, after an additional number of $D'$ time steps, each node will determine whether convergence has been achieved, and thus it will cease transmissions. 
As a result, during the operation of Algorithm~\ref{algorithm_averagedegree}, after $k_0 \geq (n-1) \tau' D + (n-1) \tau'' D + D'$ time steps, 
%where $\tau'$ fulfills \eqref{tau_value} and $\tau''$ fulfills \eqref{one_tau_value}, 
we have that each node addresses problem \textbf{P1} in Section~\ref{sec:probForm} with probability at least $p_0$. 
%\end{proof}

\section{Proof of Theorem~\ref{thm:size_estimation}}
\label{appendix:B}

We first consider Lemma~\ref{Lemma_leader}, which is necessary for our subsequent development. 

\begin{lemma}
\label{Lemma_leader}
Consider a strongly connected digraph $\mathcal{G}_d = (\mathcal{V}, \mathcal{E})$ with $n=|\mathcal{V}|$ nodes and $m=|\mathcal{E}|$ edges. 
Each node executes Algorithm~\ref{algorithm_LeaderElection}. 
After $U_v D$ time steps, a single node $v_j$ is the leader and every other node $v_i \in \mathcal{V} \setminus \{ v_j \}$ is a follower, with probability {that goes to one, as $U_v$ goes to infinity}. 
\end{lemma}

\begin{proof}
When node $v_j$ selects a value randomly depending on the number of bits allocated for communication, it basically samples from a random variable $X_j$ with probability mass function a discrete uniform distribution on the integers $0,1,2,\ldots,M-1$, where $M\coloneqq 2^{\mathrm{bits}}$. 
Let $X_1, X_2, \ldots, X_n$ be independent identically distributed (i.i.d.) random variables (representing the random variables of the $n$ nodes in the network). 
% with probability mass function a discrete uniform distribution on the integers $0,1,2,\ldots,M-1$. 
%We define the following random variable: $Z=\max\{X_1, X_2, \ldots, X_n\}$
%The probability that $X_1$ and $X_2$, at least, select the same maximum value is given by 
%\begin{align*}
%\mathbb{P}&\{X_1=t,X_2=t,X_3 \leq t, \ldots, X_n \leq t\} \\
%&= \mathbb{P}\{X_1=t\}\mathbb{P}\{X_2=t\}\prod_{i=3}^n \mathbb{P}\{X_i\leq t\}.    
%\end{align*}

Let $Y=\max\{X_1,\ldots,X_n\}$. %Then, $P[Y=x]$ is given by
%\begin{align*}
%P[Y=x] &= P[Y\leq x] - P[Y\leq x-1] \; ,~x\in\{0,\ldots,M-1\} \\  
%& = \left(\frac{x+1}{M}\right)^n -  \left(\frac{x}{M}\right)^n \; .
%\end{align*}
The probability of event $A_{\ell}$ being that $\ell<n$ nodes have the maximum value is given by 
\begin{align*}
P[A_{\ell}] 
%&= {n\choose \ell} P[X_{j_1}=Y\cap X_{j_2}=Y\cap\ldots \cap X_{j_\ell}=Y] \\  
%& = \sum_{x=0}^{M-1} P[X_{j_1}=Y\cap\ldots \cap X_{j_\ell}=Y|Y=x] P[Y=x]\\
%& ={n\choose \ell}\sum_{x=0}^{M-1} \left(\prod_{m=1}^{\ell}P[X_{j_m}=x] \right)P[Y=x] =
= {n\choose \ell}\left(\frac{1}{M}\right)^{\ell}\left(\frac{Y-1}{M}\right)^{n-\ell} \;.
\end{align*}
%Hence, the probability of event $E$ being that at least one node is being eliminated from the process is $P[E]=1-1/M^n$. 
%The probability of event $B_1$ being that $\ell_1<n$ nodes are eliminated in round $1$ is, therefore, $P[B_1]=1/M^{n-\ell_1}$. The probability of event $B_2$ being that $\ell_2<n-\ell_1$ nodes are eliminated in round $2$ is, therefore,
%\begin{align*}
%P[B_2] &=\sum_{\ell_1=0}^{n-1}P[B_2|B_1]P[B_1] \\  
%\end{align*}
%
%
%The probability that any two nodes choose the same maximum integer value is upper bounded by the probability that any two nodes choose the same integer value. 
%Therefore, the probability that all $n$ nodes involved in a round select the same value is $M \times (1/M)^n = (1/M)^{n-1}$. 
%Thus, the probability that a node is eliminated at the first round is  $1-(1/M)^{n-1}$. 
Let $n(j)$ denote the number of nodes participating in the $\max-$consensus algorithm in round $j$, $j\in\{1,2, \ldots, U_v\}$. In the worst case (when $n(j)=2$),  a node is eliminated at round $j$ with probability at least $1-1/M$. Now, we proceed with a very conservative analysis to prove Lemma~\ref{Lemma_leader} (an extended and less conservative analysis will be given in the journal version of the paper). Consider a sequence of rounds $\{1, 2, \ldots, U_v\}$: at each round a node is eliminated with probability (at least) $p=1-1/M$. Otherwise, with probability (less than) $1-p=1/M$, no node is eliminated. Thus, the probability that we have a leader after $U_v$ rounds is the probability that we have $n-1$ eliminations: this is bounded from below by 
\begin{align*}
\sum_{k=n-1}^{U_v} {U_v\choose k} & p^{k} (1-p)^{U_v-k}= 1 - \sum_{k=0}^{n-2} {U_v\choose k} p^{k} (1-p)^{U_v-k} \\
& > 1 - (n-1) {U_v\choose n-1} p^{n-1} (1-p)^{U_v-(n-1)} \;,
\end{align*}
where $U_v$ is assumed to be larger than $2(n-1)$ for the last inequality to hold (but $U_v$ is not required to be even relevant with the network size, as shown in~Fig~\ref{up_4_5_8_10000iter}).

%\todo{??? Assumption 2 requires that all the nodes share the constant $U_v$. In the proof of Theorem 3, it is required that $U_v > 2(n-1)$. This implies that, in Assumption 2, the authors require that all the nodes share a type of upper bounds on the number of the nodes. I guess an external entity must give such information. Then, a natural question is: why doesn't the external entity tell the network size directly to the nodes? ???} 

%\todo{??? $U_v > 2(n-1)$ poses another issue. To ensure that this assumption is satisfied, one needs to take $U_v$ very large. Taking large $U_v$ would cause a long execution time of the Initialization - Step 2. This issue will be serious when considering a large-scale network or when the network size is very uncertain and, therefore, one needs to rely on a rough estimate of n.  ???}

%\todo{THEMIS WILL FIX BOTH PREVIOUS}

%Let $A$ denote the event that any two or more nodes choose the same maximum integer value and $B$ denote the event that any two or more nodes choose the same value (if more than two values are chosen by many nodes, we choose the one with the highest number of nodes that selected that value). Then, $\mathbb{P}(A) < \mathbb{P}(B)$, where $\mathbb{P}(B)$ is given by
%\begin{align*}
%\mathbb{P}(B) = 1- \mathbb{P}(B^c) = 1 - \frac{^MP_n}{M^n},
%\end{align*}
%where $^MP_n\coloneqq \frac{M!}{(M-n)!}$ is the permutation and $B^c$ is the complementary of event $B$. 

The major advantage of this method is that the number of nodes participating in the next $\max-$consensus is limited to the number of nodes that had the same maximum value in the preceding $\max-$consensus round. 
%The expected number of nodes participating to the next round is given by $n/M$. 
Therefore, especially for a considerably large value of $M$, the number of nodes participating in the next $\max-$consensus is much smaller. The procedure continues until there is a round with no two nodes that select the same maximum integer value (cf. Fig~\ref{up_4_5_8_10000iter}).     
%\mathbb{P}&\{X_1=t,X_2=t,X_3 \leq t, \ldots, X_n \leq t\} \\
%&= \mathbb{P}\{X_1=t\}\mathbb{P}\{X_2=t\}\prod_{i=3}^n \mathbb{P}\{X_i\leq t\}.    
\end{proof}

Once Algorithm~\ref{algorithm_networksize} finishes the Initialization steps and elects a leader node (see Algorithm~\ref{algorithm_LeaderElection}) with high probability (see Lemma~\ref{Lemma_leader}), it executes its Iteration steps. 
However, the Iteration steps of Algorithm~\ref{algorithm_networksize} are identical to Algorithm~\ref{algorithm_averagedegree}. 
Thus the proof of Theorem~\ref{thm:size_estimation} is similar to Theorem~\ref{thm:degree_estimation} and is omitted. 

\vspace{-.3cm}

% ------------------------------------------------------------------------------
% Bibliography
% ------------------------------------------------------------------------------
\bibliographystyle{IEEEtran}
\bibliography{bibliography}

\end{document}